\documentclass[11pt]{amsart}
\usepackage[left=1.45in, right=1.4in, footskip=.25in]{geometry}
\usepackage[T1]{fontenc}
\usepackage[american]{babel}
\usepackage[europeanresistors, cuteinductors]{circuitikz} 
\usepackage{amsfonts,wasysym,bm}
\usepackage{mathrsfs}
\usepackage{caption}
\usepackage{enumerate}

\usepackage{float}
\usepackage{enumitem}
\usepackage{courier}
\usepackage[font=small]{caption}
\definecolor{mygreen}{RGB}{25,127,25}

\newtheorem{theorem}{Theorem}[section]
\newtheorem{proposition}[theorem]{Proposition}
\newtheorem{corollary}[theorem]{Corollary}
\newtheorem{lemma}[theorem]{Lemma}
\newtheorem*{thm*}{Theorem}
\newtheorem*{lem*}{Lemma}
\newtheorem*{prop*}{Proposition}

\theoremstyle{definition}
\newtheorem{definition}{Definition}[section]
\theoremstyle{remark}
\newtheorem{remark}{Remark}[section]

\theoremstyle{example}

\numberwithin{equation}{section}
\newcommand{\mbbR}{{\mathbb R}}
\newcommand{\mbbN}{{\mathbb N}}
\newcommand{\BbC}{{\mathbb C}}

\newcommand{\Hm}{\mathcal{H}}

\newcommand{\G}{\mathcal{G}}

\newcommand{\net}{\mathcal{Z}}
\DeclareMathOperator{\Pdf}{P}
\DeclareMathOperator{\Pd}{\mathcal {P}}
\DeclareMathOperator{\proj}{\pi}
\DeclareMathOperator{\SL}{FS}
\DeclareMathOperator{\Peno}{\mathcal{P}_{\net_{\varepsilon,n}}}
\newcommand{\supp}{\operatorname{supp}}
\newcommand{\dom}{\operatorname{dom}}
\newcommand{\diam}{\operatorname{diam}}
\newcommand{\emf}{\operatorname{emf}}
\newcommand{\osc}{\operatorname{osc}}

\providecommand{\pd}[1]{\Pd_{#1}}
\providecommand{\pdf}[1]{\Pdf_{#1}}
\providecommand{\Zeff}[1]{Z^{\text{eff}}_{#1}}
\title{Power dissipation in fractal Feynman-Sierpinski AC circuits}
\author{Patricia Alonso Ruiz}
\address{Department of Mathematics, University of Connecticut, Storrs, CT 06269}
\email{patricia.alonso-ruiz@uconn.edu}
\subjclass[2010]{28A80; 31C45; 94C05; 78A25}
\keywords{fractal network, power dissipation, harmonic functions, singular measure}
\thanks{This research was partly carried out with the support of the NSF grant DMS-1613025 and the Feodor-Lynen Fellowship program from the Alexander von Humboldt Foundation.}

\begin{document}
	
\begin{abstract}
This paper studies the concept of power dissipation in infinite graphs and fractals associated with passive linear networks consisting of non-dissipative elements. In particular, we analyze the so-called Feynman-Sierpinski ladder, a fractal AC circuit motivated by Feynman's infinite ladder, that exhibits power dissipation and wave propagation for some frequencies. Power dissipation in this circuit is obtained as a limit of quadratic forms, and the corresponding power dissipation measure associated with harmonic potentials is constructed. The latter measure is proved to be continuous and singular with respect to an appropriate Hausdorff measure defined on the fractal dust of nodes of the network.
\end{abstract}
\maketitle
\section{Introduction}
Passive linear networks have a wide range of applications, and especially electrical circuits have since long been intensively studied in different research areas such as electrical engineering~\cite{AV73,Bud74}, physics~\cite{Bru31,BD49} and mathematics~\cite{Wey23,Sma72}. In particular for the latter, Dirichlet forms on finite sets and graphs can be interpreted in terms of electric linear networks by considering the current flow between nodes (vertices) connected by resistors (edges). This is the key idea behind the theory of Dirichlet and resistance forms on fractals introduced by Kigami~\cite{Kig01}. In this context, one may associate these forms with ``fractal networks''.

\medskip

Resistors are just one type of passive components, or \textit{impedances}, of an electrical network. Impedances are characterized by the fact that they produce no energy by themselves. A resistor is a \textit{dissipative} element because power is lost (energy is absorbed) when an alternating current runs through it. On the contrary, no loss is caused when the current flows through a non-dissipative element such as an inductor or a capacitor. Finite linear networks consisting only of inductors and capacitors are uninteresting since no power dissipation is expected.

\medskip

However, what if the network is \textit{infinite} (as for instance fractal networks are)?

\medskip

In the 60s, Feynman posed this ``amusing question'', see~\cite[Section 22.6]{FLS64}; to give an answer, he constructed an \textit{infinite ladder network} as depicted in Figure~\ref{fig IL}. He found its behavior surprising and noticed a very interesting connection with wave propagation: 
depending on the driving frequency of the signal, power will either dissipate, allowing waves to propagate along the network, or it will not dissipate at all, preventing waves from getting through. As a consequence, voltage will either stay constant, merely changing its phase, or it will die away rapidly.
This particular infinite ladder network is what is called a \textit{low-pass filter} because low frequencies ``pass'' while higher frequencies are ``rejected''. Although such an infinite network cannot actually occur, it is often possible to realize fairly good approximations that have many technical applications, see e.g.~\cite[Section 22.7]{FLS64}.

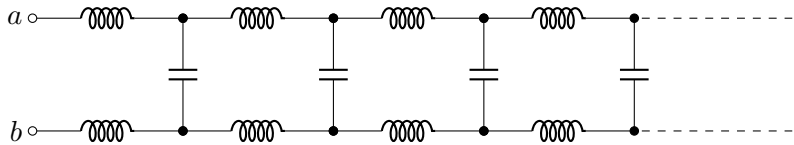
\begin{figure}[H]
\centering
\begin{tikzpicture}
\ctikzset {label/align = straight, bipoles/length=.35in}
\draw ($(0:0)$) node[left] {$b$} to[L, mirror,o-*] ++ ($(0:2)$) to[L, mirror, *-*] ++ ($(0:2)$)  to[L, mirror,*-*] ++ ($(0:2)$) to[L, mirror,*-*] ++ ($(0:2)$);
\draw[dashed]  ($(0:8.1)$) to ++ ($(0:2)$);
\draw ($(90:1.5)$) node[left] {$a$} to[L, mirror,o-*] ++ ($(0:2)$) to[L, mirror,*-*] ++ ($(0:2)$) to[L, mirror,*-*] ++ ($(0:2)$) to[L, mirror,*-*] ++ ($(0:2)$);
\draw[dashed]  ($(0:8.1)+(90:1.5)$) to ++ ($(0:2)$);
\ctikzset {label/align = straight, bipoles/length=.25in}
\draw ($(0:2)$) to[C] ++($(90:1.5)$) ($(0:4)$) to[C] ++($(90:1.5)$) ($(0:6)$) to[C] ++($(90:1.5)$) ($(0:8)$) to[C] ++($(90:1.5)$);
\end{tikzpicture}
\caption{Feynman's infinite ladder}
\label{fig IL}
\end{figure}

Also fractal structures present unusual, a priori unexpected, physical properties~\cite{ABD+12,ADL14}. Feynman's example motivated in~\cite{A++16} the construction of the so-called \textit{Feynman-Sierpinski ladder} (\textit{F-S ladder} for short), see Figure~\ref{fig SL circuit construction}, as a first prototype of a fractal network consisting solely of inductors and capacitors that exhibits power dissipation, and hence wave propagation, at some frequencies.

\medskip

The present paper aims to set up the mathematical framework to study the concept of power dissipation in infinite graphs and fractals, working out in detail the case of the F-S ladder. One of the main novelties lies in the fact that passive linear networks are studied in the frequency domain and this requires voltage, current and impedances to be considered as complex quantities. Following the classical intrinsic approach from analysis on fractals, the power dissipation in the F-S ladder will be defined as the limit of a suitable sequence of quadratic forms over complex-valued functions on its finite graph approximations.

\medskip

A crucial role in this definition is played by the \textit{harmonic functions} on the fractal dust that represents the nodes of the network. These functions describe the equilibrium potentials in a circuit when a signal is applied to the boundary nodes and they guarantee the existence of the aforementioned limit. Furthermore, proving them to be continuous will allow us to fully define power dissipation for harmonic potentials, as well as to construct the power dissipation measure associated with them. The latter measure will turn out to be singular with respect to an appropriate Hausdorff measure defined on the fractal dust related to the network.

\medskip

The paper is organized as follows: In section~\ref{section background}, we review the classical notion of power dissipation in electric passive linear networks and transfer it to graphs and infinite networks. Here, we recall the construction of the F-S ladder and set the first step towards the definition of power dissipation in this network. Section~\ref{section geometry SL} discusses some geometric properties of the projection of the F-S ladder onto $\mbbR^2$. In particular, we prove this set to be a fractal quantum graph. In order to complete the definition of power dissipation, we analyze in section~\ref{setion cp in SL} the harmonic potentials and prove in Theorem~\ref{thm harmonic cont} that they are continuous functions on the (fractal) set of nodes of the F-S ladder. Finally, section~\ref{section cpdm} deals with the construction of a measure associated with power dissipation for harmonic potentials, c.f. Theorem~\ref{thm pd is a measure}. Further, we prove in Theorem~\ref{thm singularity of pd} that this measure is singular with respect to a suitable Hausdorff measure on the set of nodes of the F-S ladder.

\section{Background and preliminaries}\label{section background}

\subsection{Complex AC currents and power dissipation}
Electric linear networks are characterized by the well-known \textit{Ohm's relation} $V=R\cdot I$, where $V$ denotes voltage, $I$ current and $R$ electric resistance. In general, passive linear networks are studied through Fourier transforms in the so-called frequency domain, requiring voltage and current to be time dependent. They are typically considered as complex quantities given by
\[
V(t)=\hat{V}e^{i(\omega t-\varphi_V)},\qquad I(t)=\hat{I}e^{i(\omega t-\varphi_I)},
\]
where $\omega$ denotes the \textit{frequency} and $\varphi_I-\varphi_V$ is the \textit{phase shift} or \textit{phase difference}. The corresponding generalization of Ohm's relation~\cite{Bru31} now becomes
\begin{equation}\label{eq ohms law}
V(t)=Z\cdot I(t),
\end{equation}
where $Z=|Z|e^{i(\varphi_I-\varphi_V)}$ is called \textit{impedance} and represents a ``complex resistance'' whose absolute value depends on the frequency $\omega$. For ease of the notation, we will assume $\varphi_V=0$ and write $\varphi:=\varphi_I$ for the phase shift, so that $Z=|Z|e^{i\varphi}$.

\medskip

Due to Kirchhoff's rules, see e.g.~\cite[Section 22.4]{FLS64}, the \textit{electromotive force} of a generator connected to a linear circuit of several impedances satisfies that 
\[
\emf(t)=Z^{\text{eff}}\cdot I(t)=|Z^{\text{eff}}|e^{i\varphi}I(t)=|Z^{\text{eff}}|\hat{I}e^{i\omega t},
\]
where $Z^{\text{eff}}$ is the so-called \textit{effective/characteristic impedance}, which represents an im\-pe\-dance equivalent to the initial set. 
The \textit{power dissipation}
, also called energy dissipation or average rate of energy loss, is given by
\begin{align*}
P&=\frac{1}{T}\int_0^T\Re(\emf(t))\Re(I(t))\,dt\\
&=\frac{1}{T}\int_0^T|Z^{\text{eff}}|\hat{I}^2\cos^2(\omega t)\cos\varphi\,dt+\frac{1}{T}\int_0^T|Z^{\text{eff}}|\hat{I}^2\cos(\omega t)\sin(\omega t)\sin\varphi\,dt\\
&=\frac{1}{2}|Z^{\text{eff}}|\hat{I}^2\cos\varphi=\frac{1}{2}|I(t)|^2\Re(Z^{\text{eff}}).
\end{align*}
Notice that this quantity only depends on the real part of the effective impedance. Consequently, a purely complex impedance is called a \textit{non-dissipative} element because, on average, there is no loss of electrical power (energy) when an alternating current runs through it.  
In order to have non-trivial power dissipation, 
it is thus necessary that the effective impedance of the circuit has positive real part. Motivated by the physical concept of power dissipation, we introduce next a quadratic form on graphs resembling this phenomenon.

\subsection{Power dissipation in finite graphs}

Let us start by considering a simple graph with two vertices $x,y$ joined by an edge $\{x,y\}$ and a network $\net$ consisting in a single impedance $Z_{xy}$ with nonzero real part. In this case, $Z_{xy}$ coincides with the effective impedance of the circuit.

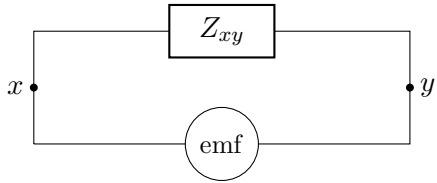
\begin{figure}[H]
\centering
\begin{tikzpicture}
[lab/.style={draw, circle, fill=white}, rect/.style={draw,rectangle,thick,minimum width=.55in,minimum height=.275in, fill=white}, ]
\coordinate (emf) at (0,0);
\coordinate (lx) at (-2.5,0);
\coordinate (rx) at (-2.5,1.5);
\coordinate[label= left:{$x$}] (x) at (-2.5,0.75);
\fill (x) circle (1.5pt);
\coordinate (Z) at (0,1.5);
\coordinate (ry) at (2.5,1.5);
\coordinate[label= right:{$y$}] (y) at (2.5,0.75);
\fill (y) circle (1.5pt);
\coordinate (ly) at (2.5,0);
\draw (lx) -- (rx) -- node[midway, rect] {$Z_{xy}$} (ry) -- (ly) -- node[midway, lab,] {\small{$\emf$}} (lx);
\end{tikzpicture}
\caption{Basic $1$-edge network.}
\label{simple network}
\end{figure}

The voltage across the edge $\{x,y\}$ is defined as the difference between the potential at each node. Following Ohm's law~\eqref{eq ohms law}, for any potential function $(v(x),v(y))\in\BbC^2$ the current flowing from $x$ to $y$ is given by
\begin{equation}\label{eq new Ohm}
I_{xy}=\frac{v(y)-v(x)}{Z_{xy}}.
\end{equation}
Although these quantities actually depend on time, we will consider this parameter fixed and omit it hereafter in the whole discussion.
\begin{definition}
The power dissipation associated with the network $\net=\{Z_{xy}\}$ is defined as the quadratic form $\pd{\omega,\net}\colon\BbC^2\to\mbbR$ given by
\begin{equation}\label{eq P(v)}
\pd{\omega,\net}[v]_{xy}=\frac{1}{2}\frac{\Re(Z_{xy})}{|Z_{xy}|^2}|v(x)-v(y)|^2,
\end{equation}
equivalently
\[
\pd{\omega,\net}[v]_{xy}=\frac{\cos\varphi_{xy}}{2|Z_{xy}|}|v(x)-v(y)|^2,
\]
where $\varphi_{xy}$ is the phase shift associated with $Z_{xy}$.
\end{definition}
The subindex $\omega$ refers to the dependence of this expression on the frequency. For ease of the reading we will drop it off in the sequel and refer to this dependence explicitly only when confusion may occur.
\medskip

%
Once power dissipation is defined for functions on a simple ($2$-node) network, this concept naturally extends to networks in graphs with several vertices and edges. In order for~\eqref{eq new Ohm} to provide a one-to-one correspondence between potentials and currents and thus identify functions on edges with functions on vertices, we will restrict our discussion to graphs/networks without multiple edges.
\medskip

Let us now consider a finite graph $\G=(V,E)$ and a network $\net=\{Z_{xy}~|~\{x,y\}\in E\}$ consisting of impedances $Z_{xy}$ attached to each edge $\{x,y\}$. Further, we denote by $\ell(V)$ the space of complex-valued functions on $V$. 
\begin{definition}
The quadratic form $\pd{\net}\colon\ell(V)\to\mbbR$ given by
\begin{equation}\label{eq general def P(v)}
\pd{\net}[v]:=\sum_{\{x,y\}\in E}\pd{\net}[v]_{xy}
\end{equation}
is called the \textit{power dissipation} in $\G$ associated with the network $\net$. 
\end{definition}
\begin{remark}
In the case of time-independent circuits with purely real impedances (resistors), power dissipation coincides with the classical definition of a \textit{resistance/energy form}. Indeed, if the quantities $Z_{x,y}$, $I_{xy}$ and the function $v$ are real,~\eqref{eq P(v)} becomes
\[
\pd{\net}[v]=\frac{1}{2}\sum_{\{x,y\}\in E}\frac{1}{Z_{xy}}(v(x)-v(y))^2,
\]
where $Z_{xy}$ is the resistance between $x$ and $y$. 
\end{remark}

In the next section we extend the notion of power dissipation to infinite graphs and networks, focusing on the particular case of the \textit{Feynman-Sierpinski ladder} (F-S ladder), for which some useful computations and results have been obtained in~\cite{A++16}.

\subsection{Power dissipation in infinite networks. The F-S ladder}\label{section PD in infinite networks}
The F-S ladder circuit that we denote by $\net_{\SL}$ was introduced in~\cite{A++16} as a fractal network whose underlying graph structure is described in Figure~\ref{fig graphs G_n}. The infinite graph that arises in the limit can be formally embedded in $\mbbR^2$; existence and geometric properties of this set are discussed in Section~\ref{section geometry SL}.

\begin{figure}[H]\centering
\begin{tabular}{cccc}
\begin{tikzpicture}[scale=2/3]
\draw($(90:3)$) -- ($(90+120:3)$) -- node[midway, below] {$\G_0$} ($(90+120*2:3)$) -- cycle;\end{tikzpicture}
&
\begin{tikzpicture}[scale=2/3]
\foreach \a in {0,1,2} {
\draw ($(90+120*\a:1)+(90:1)$) -- ($(90+120*\a:1)+(210:1)$) -- ($(90+120*\a:1)+(330: 1)$)--cycle; 
}
\draw ($(90:2)$) -- ($(90:3)$) -- ($(90+120:3)$)  ($(90+120:2)$) -- ($(90+120:3)$)  -- ($(90+120*2:3)$) node[midway, below] {$\G_1$} ($(90+120*2:2)$) --($(90+120*2:3)$) -- ($(90:3)$);

\end{tikzpicture}
&
\begin{tikzpicture}[scale=2/9]
\foreach \b in {0,1,2} {
\foreach \a in {0,1,2} {

\draw ($(90+120*\b:3)+(90+120*\a:1)+(90:1)$) -- ($(90+120*\b:3)+(90+120*\a:1)+(210:1)$) -- ($(90+120*\b:3)+(90+120*\a:1)+(330: 1)$)--cycle; 
}
\draw ($(90+120*\b:3)+(90:2)$) -- ($(90+120*\b:3)+(90:3)$) -- ($(90+120*\b:3)+(90+120:3)$)  ($(90+120*\b:3)+(90+120:2)$) -- ($(90+120*\b:3)+(90+120:3)$)  -- ($(90+120*\b:3)+(90+120*2:3)$)  ($(90+120*\b:3)+(90+120*2:2)$) --($(90+120*\b:3)+(90+120*2:3)$) -- ($(90+120*\b:3)+(90:3)$);
}
\draw ($(90:6)$) -- ($(90:9)$) -- ($(90+120:9)$)  ($(90+120:6)$) -- ($(90+120:9)$) -- ($(90+120*2:9)$) node[midway, below] {$\G_2$} ($(90+120*2:6)$) --($(90+120*2:9)$) -- ($(90:9)$);
\end{tikzpicture}
&
\begin{tikzpicture}[scale=2/27]
\foreach \c in {0,1,2} {
\foreach \b in {0,1,2} {
\foreach \a in {0,1,2} {

\draw ($(90+120*\c:9)+(90+120*\b:3)+(90+120*\a:1)+(90:1)$) -- ($(90+120*\c:9)+(90+120*\b:3)+(90+120*\a:1)+(210:1)$) -- ($(90+120*\c:9)+(90+120*\b:3)+(90+120*\a:1)+(330:1)$)--cycle; 
}
\draw ($(90+120*\c:9)+(90+120*\b:3)+(90:2)$) -- ($(90+120*\c:9)+(90+120*\b:3)+(90:3)$) -- ($(90+120*\c:9)+(90+120*\b:3)+(90+120:3)$)  ($(90+120*\c:9)+(90+120*\b:3)+(90+120:2)$) -- ($(90+120*\c:9)+(90+120*\b:3)+(90+120:3)$)  -- ($(90+120*\c:9)+(90+120*\b:3)+(90+120*2:3)$)  ($(90+120*\c:9)+(90+120*\b:3)+(90+120*2:2)$) --($(90+120*\c:9)+(90+120*\b:3)+(90+120*2:3)$) -- ($(90+120*\c:9)+(90+120*\b:3)+(90:3)$);
}
\draw ($(90+120*\c:9)+(90:6)$) -- ($(90+120*\c:9)+(90:9)$) -- ($(90+120*\c:9)+(90+120:9)$) ($(90+120*\c:9)+(90+120:6)$) -- ($(90+120*\c:9)+(90+120:9)$)  -- ($(90+120*\c:9)+(90+120*2:9)$)  ($(90+120*\c:9)+(90+120*2:6)$) --($(90+120*\c:9)+(90+120*2:9)$) -- ($(90+120*\c:9)+(90:9)$);
}
\draw ($(90:18)$) -- ($(90:27)$) -- ($(90+120:27)$)  ($(90+120:18)$) -- ($(90+120:27)$)  -- ($(90+120*2:27)$) node[midway, below] {$\G_3$} ($(90+120*2:18)$) --($(90+120*2:27)$) -- ($(90:27)$);
\end{tikzpicture}
\end{tabular}
\caption{Approximation of $\G_\infty$ by graphs.}
\label{fig graphs G_n}
\end{figure}
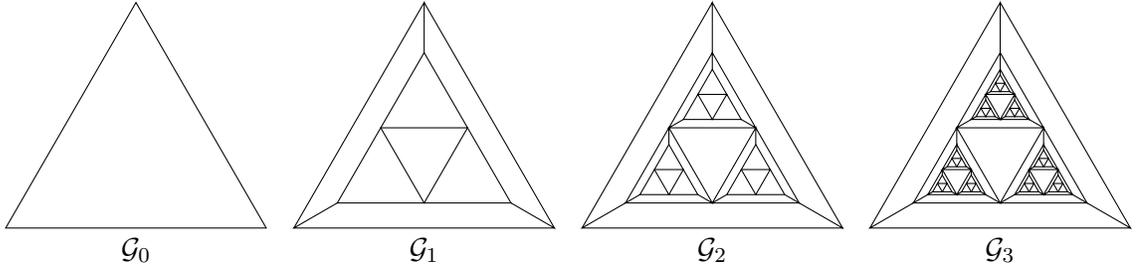
Let $\G_n=(V_n,E_n)$, $n\geq 0$, be the graphs displayed in Figure~\ref{fig graphs G_n} and let $\G_\infty$ denote the limit (in the Gromov-Hausdorff metric) of the sequence $\{\G_n\}_{n\geq 0}$. Notice that this limit exists in view of Proposition~\ref{prop existence Qalpha} and Remark~\ref{rem Qalpha}.
The F-S ladder $\net_{\SL}$ is the infinite network on $\G_\infty$ whose edges have impedances that are capacitors $Z_C=\frac{1}{i\omega C}$ or inductors $Z_L=i\omega L$ with $C,L>0$ as shown in Figure~\ref{fig SL circuit construction}. By convention, the symbol $\dashv\,\vdash$ is employed for capacitors, and $\gluon$ for inductors. 
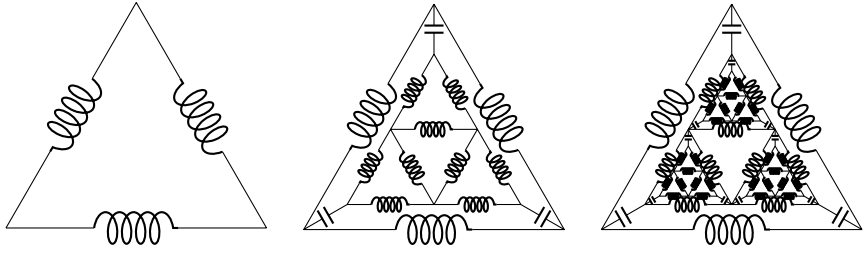
\begin{figure}[H]
\centering
\begin{tabular}{ccc}
\begin{circuitikz}[scale=2/3]
\draw ($(90:3)$)  to[L, mirror] ($(90+120:3)$) to[L, mirror] ($(90+120*2:3)$) to[L, mirror] ($(90:3)$);
\end{circuitikz}
&
\begin{circuitikz}[scale=2/3]
\foreach \a in {0,1,2} {
\ctikzset {label/align = straight, bipoles/length=.25in}
\draw ($(90+120*\a:1)+(90:1)$) to[L, mirror] ($(90+120*\a:1)+(210:1)$) to[L, mirror] ($(90+120*\a:1)+(330: 1)$) to[L, mirror] ($(90+120*\a:1)+(90:1)$); 
}
\ctikzset {label/align = straight, bipoles/length=.5in, bipoles/capacitor/height=.2, bipoles/capacitor/width=.075}
\draw ($(90:2)$) to[C] ($(90:3)$) to[L,  mirror] ($(90+120:3)$)  ($(90+120:2)$) to[C] ($(90+120:3)$) to[L,  mirror] ($(90+120*2:3)$) ($(90+120*2:2)$) to[C] ($(90+120*2:3)$) to[L, mirror] ($(90:3)$);
\end{circuitikz}
&
\begin{circuitikz}[scale=2/9]
\foreach \b in {0,1,2} {
\foreach \a in {0,1,2} {
\ctikzset {label/align = straight, bipoles/length=.1in}
\draw ($(90+120*\b:3)+(90+120*\a:1)+(90:1)$) to[L,  mirror] ($(90+120*\b:3)+(90+120*\a:1)+(210:1)$) to[L,  mirror] ($(90+120*\b:3)+(90+120*\a:1)+(330: 1)$) to[L,  mirror] ($(90+120*\b:3)+(90+120*\a:1)+(90:1)$); 
}
\ctikzset {label/align = straight, bipoles/length=.25in, bipoles/capacitor/height=.15, bipoles/capacitor/width=.075}
\draw ($(90+120*\b:3)+(90:2)$) to[C] ($(90+120*\b:3)+(90:3)$);
\draw ($(90+120*\b:3)+(90:3)$) to[L, mirror] ($(90+120*\b:3)+(90+120:3)$)  ($(90+120*\b:3)+(90+120:2)$) to[C] ($(90+120*\b:3)+(90+120:3)$)  to[L, mirror] ($(90+120*\b:3)+(90+120*2:3)$)  ($(90+120*\b:3)+(90+120*2:2)$) to[C] ($(90+120*\b:3)+(90+120*2:3)$) to[L, mirror] ($(90+120*\b:3)+(90:3)$);
}
\ctikzset {label/align = straight, bipoles/length=.5in, bipoles/capacitor/height=.2, bipoles/capacitor/width=.075}
\draw ($(90:6)$) to[C] ($(90:9)$) to[L, mirror] ($(90+120:9)$)  ($(90+120:6)$) to[C] ($(90+120:9)$) to[L, mirror] ($(90+120*2:9)$)  ($(90+120*2:6)$) to[C] ($(90+120*2:9)$) to[L, mirror] ($(90:9)$);
\end{circuitikz}
\end{tabular}
\caption{Construction of the Feynman-Sierpinski ladder circuit.}
\label{fig SL circuit construction}
\end{figure}

As for any passive linear network, power dissipation in $\net_{\SL}$ is only meaningful if its effective impedance $\Zeff{\SL}$, that depends on the frequency $\omega$, has positive real part. In the case of the F-S ladder, this is satisfied under the \textit{filter condition}
\begin{equation}\label{FC}
9(4-\sqrt{15})<2\omega^2LC<9(4+\sqrt{15}),
\end{equation}
c.f.~\cite[Theorem 3.1]{A++16}.

\medskip

Following the intrinsic approach from analysis on fractals, we introduce next a sequence of networks on the approximating graphs $\G_n$ that will eventually lead to the definition of the power dissipation in $\G_\infty$ associated with $\net_{\SL}$.
\subsection*{Networks $\net_{\varepsilon,n}=\{Z_{\varepsilon,n,xy}~|~\{x,y\}\in E_n\}$}
At each level $n\geq 1$, the network $\net_{\varepsilon,n}$ is constructed by adding a small positive resistance $\varepsilon$ in series with each of the impedances of $\net_{\SL}$, see Figure~\ref{SL epsilon approx}.  Thus, for each $\{x,y\}\in E_n$,
\begin{equation}\label{eq def Zepsilon}
Z_{\varepsilon,n,xy}=Z_{\SL,xy}+\varepsilon,
\end{equation}
where $Z_{\SL,xy}\in\{Z_C, Z_L\}$ according to the previous construction in Figure~\ref{fig SL circuit construction}. Under the filter condition~\eqref{FC}, we know from~\cite[Theorem 3.2]{A++16} that the effective impedance of the F-S ladder, $\Zeff{\SL}$, is the regularized limit of the effective impedances of the networks $\{\net_{\varepsilon,n}\}_{n\geq 1}$, i.e.
\begin{equation}\label{eq reg limit}
\Zeff{\SL}=\lim_{\varepsilon\to 0_+}\lim_{n\to\infty}\Zeff{\varepsilon,n}.
\end{equation}

Furthermore, we set $Z_{\varepsilon,0,xy}=\Zeff{\varepsilon}:=\lim\limits_{n\to\infty}\Zeff{\varepsilon,n}$ for all $\{x,y\}\in E_0$, and from now on assume that the F-S ladder satisfies the filter condition~\eqref{FC}.

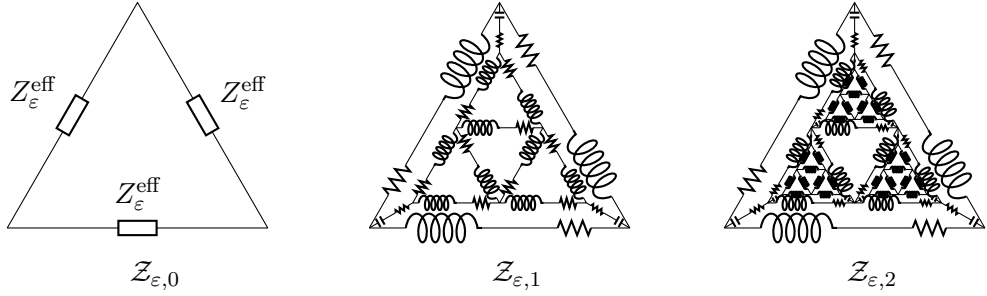
\begin{figure}[H]
\centering
\begin{tabular}{ccc}
\begin{circuitikz}[scale=2/3]
\ctikzset {label/align = straight, bipoles/length=.25in, bipoles/resistor/height=.15, bipoles/resistor/width=.2}
\draw ($(90:3)$) to[R=$\Zeff{\varepsilon}$, mirror] ($(90+120:3)$) to[R=$\Zeff{\varepsilon}$] ($(90+120*2:3)$)to[R=$\Zeff{\varepsilon}$, mirror] ($(90:3)$);
\node[below] at ($(280:2)$) {$\net_{\varepsilon,0}$};
\end{circuitikz}
\hspace*{.25in}
&
\begin{circuitikz}[scale=2/3]
\foreach \a in {0,1,2} {
\ctikzset {label/align = straight, bipoles/length=.25in, bipoles/capacitor/height=.2, bipoles/capacitor/width=.075,bipoles/resistor/height=.2, bipoles/resistor/width=.35}
\draw ($(90+120*\a:1)+(90:1)$) to[L, mirror] ++($(90+150:1)$) to[american resistor] ($(90+120*\a:1)+(210:1)$) to[L, mirror] ++($(0:1)$) to[american resistor] ($(90+120*\a:1)+(330: 1)$) to[L, mirror] ++($(90+30:1)$) to[american resistor] ($(90+120*\a:1)+(90:1)$); 
}
\ctikzset {label/align = straight, bipoles/length=.5in, bipoles/capacitor/height=.1, bipoles/capacitor/width=.05,bipoles/resistor/height=.05, bipoles/resistor/width=.15}
\draw ($(90:2)$) to[american resistor] ++($(90:3ex)$) to[C] ($(90:3)$) ($(90+120:2)$) to[american resistor] ++($(90+120:3ex)$) to[C] ($(90+120:3)$) ($(90+120*2:2)$) to[american resistor] ++($(90+120*2:3ex)$) to[C] ($(90+120*2:3)$);

\ctikzset {label/align = straight, bipoles/length=.5in,bipoles/resistor/height=.15, bipoles/resistor/width=.35}
\draw ($(90:3)$) to[L,  mirror] ++($(90+150:18ex)$) to[american resistor] ($(90+120:3)$) to[L,  mirror] ++($(0:18ex)$) to[american resistor] ($(90+120*2:3)$) to[L, mirror] ++($(120:18ex)$) to[american resistor] ($(90:3)$);
\node[below] at ($(280:2)$) {$\net_{\varepsilon,1}$};
\end{circuitikz}
\hspace*{.25in}
&
\begin{circuitikz}[scale=2/9]
\foreach \b in {0,1,2} {
\foreach \a in {0,1,2} {
\ctikzset {label/align = straight, bipoles/length=.1in}
\draw ($(90+120*\b:3)+(90+120*\a:1)+(90:1)$) to[L, mirror] ($(90+120*\b:3)+(90+120*\a:1)+(210:1)$) to[L, mirror] ($(90+120*\b:3)+(90+120*\a:1)+(330: 1)$) to[L, mirror] ($(90+120*\b:3)+(90+120*\a:1)+(90:1)$); 
}

\ctikzset {label/align = straight, bipoles/length=.1in, bipoles/capacitor/height=.1, bipoles/capacitor/width=.075}
\draw ($(90+120*\b:3)+(90:2)$)  to[C] ($(90+120*\b:3)+(90:3)$);

\ctikzset {label/align = straight, bipoles/length=.25in,bipoles/resistor/height=.1, bipoles/resistor/width=.25}
\draw ($(90+120*\b:3)+(90:3)$) to[L, mirror] ++($(90+150:20ex)$) to[american resistor] ($(90+120*\b:3)+(90+120:3)$)  ($(90+120*\b:3)+(90+120:2)$) to[C] ($(90+120*\b:3)+(90+120:3)$)  to[L, mirror] ++($(0:20ex)$) to[american resistor] ($(90+120*\b:3)+(90+120*2:3)$)  ($(90+120*\b:3)+(90+120*2:2)$) to[C] ($(90+120*\b:3)+(90+120*2:3)$) to[L, mirror] ++($(120:20ex)$) to[american resistor] ($(90+120*\b:3)+(90:3)$);
}
\ctikzset {label/align = straight, bipoles/length=.5in, bipoles/capacitor/height=.1, bipoles/capacitor/width=.05,bipoles/resistor/height=.05, bipoles/resistor/width=.15}
\draw ($(90:6)$) to[american resistor] ++($(90:9ex)$) to[C] ($(90:9)$) ($(90+120:6)$) to[american resistor] ++($(90+120:9ex)$) to[C] ($(90+120:9)$) ($(90+120*2:6)$) to[american resistor] ++($(90+120*2:9ex)$) to[C] ($(90+120*2:9)$); 
\ctikzset {label/align = straight, bipoles/length=.5in,bipoles/resistor/height=.15, bipoles/resistor/width=.35}
\draw ($(90:9)$) to[L, mirror] ++($(90+150:54ex)$) to[american resistor] ($(90+120:9)$) ($(90+120:9)$) to[L, mirror] ++($(0:54ex)$) to[american resistor] ($(90+120*2:9)$)  ($(90+120*2:9)$) to[L, mirror] ++($(120:54ex)$) to[american resistor] ($(90:9)$);
\node[below] at ($(280:6)$) {$\net_{\varepsilon,2}$};
\end{circuitikz}
\end{tabular}
\caption{Approximating networks $\net_{\varepsilon,n}$.}
\label{SL epsilon approx}
\end{figure}

In view of~\eqref{eq reg limit}, the sequence of networks $\{\net_{\varepsilon,n}\}_{n\geq 0}$ provides the base towards the desired definition of power dissipation.

\begin{definition}\label{def pd graphs}
Let $V_*=\bigcup_{n\geq 0}V_n$. The power dissipation in $\G_\infty$ associated with the F-S ladder is the quadratic form $\pd{\SL}\colon\dom\pd{\SL}\to\mbbR$ given by
\[
\pd{\SL}[v]:=\lim_{\varepsilon\to 0_+}\lim_{n\to\infty}\Peno[v_{|_{V_n}}]
\]
and $\dom\pd{\SL}=\{v\in\ell(V_*)~|~\pd{\SL}[v]<\infty\}$.
\end{definition}

The embedding of the infinite graph $\G_\infty$ in $\mbbR^2$ presented in Section~\ref{section geometry SL} will reveal that the set of vertices of $\G_\infty$ is actually larger than $V_*$, so that the latter form is actually incomplete. Its definition for potentials defined on the whole network will appear in Section~\ref{setion cp in SL}, c.f. Definition~\ref{def pd SL}.
\begin{remark}\label{rem about Zepsn}
At each approximating level, the network $\net_{\varepsilon,n}$ is equivalent to a triangular network with impedances $\Zeff{\varepsilon,n}$. Thus, for any $n\geq 0$ and $u\in\ell(V_0)$,
\[
\min\{\Peno[v]~|~v\in\ell(V_n),v_{|_{V_0}}=u\}=\frac{\Re(\Zeff{\varepsilon,n})}{2|\Zeff{\varepsilon,n}|^2}\sum_{\{x,y\}\in E_0}|u(x)-u(y)|^2.
\]
\end{remark}

The latter remark is directly related to \textit{harmonic functions}. These describe the equilibrium states of the F-S ladder when a potential is connected to the boundary vertices of the circuit, which in this case consists of the three vertices in $V_0$. More precisely, a function $h\in\ell(V_*)$ is said to be harmonic if for any $\varepsilon>0$ and $n\geq 1$
\[
\pd{\net_{\varepsilon,0}}[h_{|_{V_0}}]=\Peno[h_{|_{V_n}}].
\]
In particular, $\pd{\net_{\SL}}[h]=\lim\limits_{\varepsilon\to 0_+}\pd{\net_{\varepsilon,n}}[h_{|_{V_n}}]$ for any $n\geq 0$. Since $V_0$ has three elements, the space of harmonic functions on $V_*$, that we denote by $\Hm_{\SL}(V_*)$, is a $3$-dimensional subspace of $\dom\pd{\SL}$. 

\medskip

In connection with harmonic functions we introduce the following auxiliary networks on the approximating graphs $\G_n$. 
\subsection*{Networks $\net_n=\{Z_{n,xy}~|~\{x,y\}\in E_n\}$}
At each $n\geq 1$, this network is constructed by changing the impedance of edges building triangles in the ``deepest approximation level'', i.e. $\{x,y\}\in E_n\setminus E_{n-1}$ with $x,y\in V_n\setminus V_{n-1}$, to equal the effective impedance of the whole network $\net_{\SL}$. The elements of $\net_n$ are thus given by
\begin{equation*}\label{eq def Zen}
Z_{n,xy}=\left\{
\begin{array}{ll}
\Zeff{\SL}& \text{if }\{x,y\}\in E_n\setminus E_{n-1},x,y\in V_n\setminus V_{n-1},\\
&\\
Z_{\SL,xy} &\text{otherwise.}
\end{array}
\right.
\end{equation*}
For completeness, we set $Z_{0,xy}=\Zeff{\SL}$ for all $\{x,y\}\in E_0$. One of the most relevant differences between $\net_n$ and $\net_{\varepsilon,n}$ is that the impedance of the edges changes with the approximation level. Moreover, the impedance of edges building triangles have non-zero real part, whereas the impedance of the remaining edges is purely imaginary.

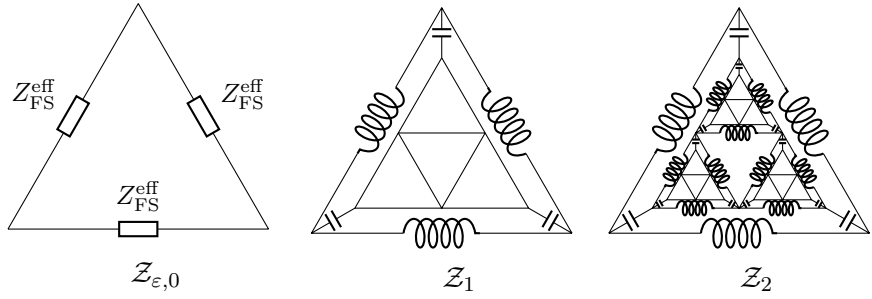
\begin{figure}[H]
\centering
\begin{tabular}{ccc}
\begin{circuitikz}[scale=2/3]
\ctikzset {label/align = straight, bipoles/length=.25in, bipoles/resistor/height=.15, bipoles/resistor/width=.2}
\draw ($(90:3)$) to[R=\small{$\Zeff{\SL}$}, mirror] ($(90+120:3)$) to[R=\small{$\Zeff{\SL}$}] ($(90+120*2:3)$)to[R=\small{$\Zeff{\SL}$}, mirror] ($(90:3)$);
\node[below] at ($(280:2)$) {$\net_{\varepsilon,0}$};
\end{circuitikz}
&
\begin{circuitikz}[scale=2/3]
\ctikzset {label/align = straight, bipoles/length=.5in, bipoles/capacitor/height=.2, bipoles/capacitor/width=.075}
\foreach \a in {0,1,2} {
\draw ($(90+120*\a:1)+(90:1)$) -- ($(90+120*\a:1)+(210:1)$) -- ($(90+120*\a:1)+(330: 1)$) -- ($(90+120*\a:1)+(90:1)$); 
}
\draw ($(90:2)$) to[C] ($(90:3)$) to[L,  mirror] ($(90+120:3)$)  ($(90+120:2)$) to[C] ($(90+120:3)$) to[L,  mirror] ($(90+120*2:3)$) ($(90+120*2:2)$) to[C] ($(90+120*2:3)$) to[L, mirror] ($(90:3)$);
\node[below] at ($(280:2)$) {$\net_1$};
\end{circuitikz}
&
\begin{circuitikz}[scale=2/9]
\foreach \b in {0,1,2} {
\ctikzset {label/align = straight, bipoles/length=.25in, bipoles/capacitor/height=.15, bipoles/capacitor/width=.075}

\foreach \a in {0,1,2} {
\ctikzset {label/align = straight, bipoles/length=.2in}
\draw ($(90+120*\b:3)+(90+120*\a:1)+(90:1)$) -- ($(90+120*\b:3)+(90+120*\a:1)+(210:1)$) -- ($(90+120*\b:3)+(90+120*\a:1)+(330: 1)$) -- ($(90+120*\b:3)+(90+120*\a:1)+(90:1)$); 
}
\draw ($(90+120*\b:3)+(90:2)$) to[C] ($(90+120*\b:3)+(90:3)$);
\draw ($(90+120*\b:3)+(90:3)$) to[L, mirror] ($(90+120*\b:3)+(90+120:3)$)  ($(90+120*\b:3)+(90+120:2)$) to[C] ($(90+120*\b:3)+(90+120:3)$)  to[L, mirror] ($(90+120*\b:3)+(90+120*2:3)$)  ($(90+120*\b:3)+(90+120*2:2)$) to[C] ($(90+120*\b:3)+(90+120*2:3)$) to[L, mirror] ($(90+120*\b:3)+(90:3)$);
}
\ctikzset {label/align = straight, bipoles/length=.5in, bipoles/capacitor/height=.2, bipoles/capacitor/width=.075}
\draw ($(90:6)$) to[C] ($(90:9)$) to[L, mirror] ($(90+120:9)$)  ($(90+120:6)$) to[C] ($(90+120:9)$) to[L, mirror] ($(90+120*2:9)$)  ($(90+120*2:6)$) to[C] ($(90+120*2:9)$) to[L, mirror] ($(90:9)$);
\node[below] at ($(280:6)$) {$\net_2$};
\end{circuitikz}
\end{tabular}
\caption{Approximating networks $\net_n$.}
\label{SL effective approx}
\end{figure}

In view of~\cite[Theorem 3.1]{A++16}, the networks $\net_n$ are all electrically equivalent. This fact relates them directly to harmonic functions and power dissipation, as the next proposition shows.
\begin{proposition}\label{prop harmonic in V}
For any $h\in\Hm_{\SL}(V_*)$ it holds that
\[
\pd{\net_{\SL}}[h]=\pd{\net_n}[h_{|_{V_n}}]\qquad~\forall~n\geq 0.
\]
\end{proposition}
Here, $\pd{\SL}[h]$ acts in place of what one would in general define as ``trace'' of the ``limit'' power dissipation.
\begin{proof}
We prove the equivalent statement that for each $n\geq 0$, measuring potential across vertices in 
the F-S ladder network is equivalent to measuring them in the network $\net_n$.

By definition of $\net_n$, the impedance between edges $\{x,y\}\in E_{n-1}$ with at least one vertex in $V_{n-1}$ is the same in $\net_{\SL}$ and $\net_n$. Thus, it suffices to show that triangular cells of level $n$ ($n$-cells) are electrically equivalent in both networks. 

On the one hand, notice that an $n$-cell of the network $\net_n$ is a triangular network with impedances $\Zeff{\SL}$. On the other hand, an $n$-cell of $\net_{\SL}$ is itself a F-S ladder, which is  electrically equivalent to a triangular network with impedances $\Zeff{\SL}$ as well.
\end{proof}
\begin{remark}\label{rem about Zen}
From the definition of~$\net_n$ it follows that for any $v\in\ell(V_*)$ and $n\geq 1$,
\[
\pd{\net_n}[v_{|_{V_n}}]=\frac{1}{2}\sum_{\substack{\{x,y\}\in E_n\setminus E_{n-1}\\x,y\in V_n\setminus V_{n-1}}}\frac{\Re(\Zeff{\SL})}{|\Zeff{\SL}|^2}|v(x)-v(y)|^2,
\]
which is a multiple of the energy of the $n$-th graph approximation of the Sierpinski gasket.
\end{remark}

\section{Geometric projection of the infinite graph $\G_\infty$}\label{section geometry SL}

The power dissipation associated with the F-S ladder has so far been defined for potentials on $V_*$. The present section investigates some geometric properties of the subset of $\mbbR^2$ that corresponds to the graphical representation of $\G_\infty$. Among them, this set turns out to be a fractal quantum graph whose set of vertices is a fractal dust larger than $V_*$. 

\medskip

Let $S=\{1,2,3\}$ and let $\{p_1,p_2, p_3\}\in\mbbR^2$ denote the set of vertices of an equilateral triangle of side length $1$ with baricentre $p_0$.

\begin{definition}\label{def G_i}
For each $i\in S$, define the map $G_i\colon\mbbR^2\to\mbbR^2$ as
\[
G_i(x):=F_i\circ G_0(x),
\]
where $F_i,G_0\colon\mbbR^2\to\mbbR^2$ are given by
\[
F_i(x)=\frac{1}{2}(x-G_0(p_i))+G_0(p_i),
\]
respectively
\[
G_0(x)=\alpha(x-p_0)+p_0,
\]
with $\alpha\in(0,1)$. Moreover, set $p_{ij}=G_i(p_j)$ for each $i,j\in S$.
\end{definition}

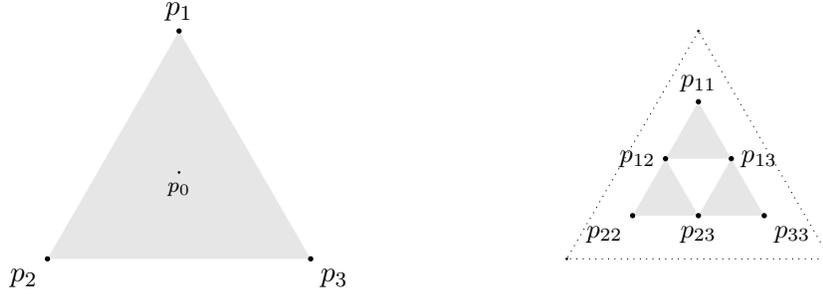
\begin{figure}[H]
\centering
\begin{tikzpicture}[scale=0.5]
\coordinate[label=below left:{$p_2$}] (p_1) at (0,0);
\coordinate[label=above:{$p_1$}] (p_2) at (3.5,6.0621);
\coordinate[label=below right:{$p_3$}] (p_3) at (7,0);
\coordinate[label=below:{$\scriptstyle{p_0}$}] (p_0) at (3.5,2.3);
\fill[very nearly transparent] (p_1) -- (p_2) -- (p_3) -- cycle;
\fill (p_0) circle (1pt);
\fill (p_1) circle (2pt);
\fill (p_2) circle (2pt);
\fill (p_3) circle (2pt);
\end{tikzpicture}
\hspace*{.75in}
\begin{tikzpicture}[scale=0.5]
\coordinate (p_0) at (3.5,2.3);
\coordinate[label=above:\textcolor{white}{$p_1$}] (p_1) at (3.5,6.0621);
\coordinate[label=below left:\textcolor{white}{$p_2$}] (p_2) at (0,0);
\coordinate[label=below right:\textcolor{white}{$p_3$}] (p_3) at (7,0);

\coordinate[label=above:{\small{$p_{11}$}}] (p_11) at (3.5/2+3.5/2,6.0621/2+2.3/2);
\coordinate[label=left:{\small{$p_{12}$}}] (p_12) at (3.5/4+3.5/2,6.0621/4+2.3/2);
\coordinate[label=right:{\small{$p_{13}$}}] (p_13) at (7/4+3.5/4+3.5/2,6.0621/4+2.3/2);
\coordinate[label=below left:{\small{$p_{22}$}}] (p_22) at (3.5/2,2.3/2);
\coordinate[label=below:{\small{$p_{23}$}}] (p_23) at (7/4+3.5/2,2.3/2);
\coordinate[label=below right:{\small{$p_{33}$}}] (p_33) at (7/2+3.5/2,2.3/2);
%
\draw[thin,dotted] (p_1) -- (p_2) -- (p_3) -- cycle;
\fill[very nearly transparent] (p_11) -- (p_12) -- (p_13) -- cycle;
\fill[very nearly transparent] (p_12) -- (p_22) -- (p_23) -- cycle;
\fill[very nearly transparent] (p_13) -- (p_23) -- (p_33) -- cycle;
\fill (p_1) circle (1pt);
\fill (p_2) circle (1pt);
\fill (p_3) circle (1pt);
\fill (p_11) circle (2pt);
\fill (p_12) circle (2pt);
\fill (p_13) circle (2pt);
\fill (p_22) circle (2pt);
\fill (p_23) circle (2pt);
\fill (p_33) circle (2pt);
\end{tikzpicture}
\caption{Shadowed: initial equilateral triangle and its image under $G_1,G_2$ and $G_3$.}
\end{figure}

Notice that $p_i$ is \textit{not} the fixed point of $G_i$ for any $i\in S$, and therefore $p_i\neq p_{ii}$ for any $i\in S$. On the other hand, since $p_{ij}=p_{ji}$ for all $i\neq j$, we will restrict ourselves to writing $p_{ij}$ only for $i\leq j$. Although the mappings $G_i$ actually depend on $\alpha$, we will see in Proposition~\ref{prop Qalpha homeo} that all lead to topologically equivalent sets, which eventually makes the parameter $\alpha$ irrelevant.

\begin{definition}
Let $W_0 = \{\emptyset\}$ and define for $n \geq 1$
\[
W_n = \{w~|~w = w_1\ldots{w_n}, w_i \in S,~i = 1, \ldots, n\}.
\]
Moreover, let $W_* = \cup_{n \geq 0} W_n$ and for any $w = w_1\ldots{w_n} \in W_*$ define $G_w\colon\mbbR^2\to\mbbR^2$ by
\[
G_w = G_{w_1}\circ G_{w_2}\circ\cdots\circ G_{w_n},
\]
setting $G_\emptyset$ to be the identity on $\mbbR^2$. Finally, define $\widetilde{V}_0 = \{p_1, p_2, p_3\}$ and
\begin{equation}\label{eq def tilde V}
\widetilde{V}_n = \bigcup_{w \in W_n} G_w(\tilde{V}_0)
\end{equation}
for $n \ge 1$, as well as $\widetilde{V}_*=\cup_{n\geq 0}\widetilde{V}_n$.
\end{definition}
For each $i\in S$ we will denote by $e_{ii}$ the line segment joining $p_i$ and $p_{ii}$, and by $e_{ij}$ the line segment joining $p_{i}$ and $p_{j}$, with $i<j$, see Figure~\ref{fig segments}. Moreover, we define $B:= \{(i,j)~|~i\leq j\}$ and write $e_{ij}^w = G_w(e_{ij})$ for any $(w,(i,j))\in W_*\times B$.
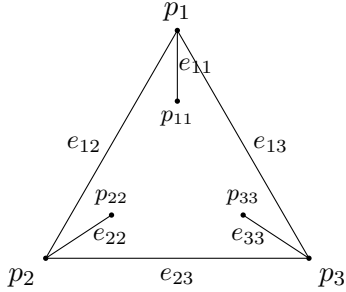
\begin{figure}[H]
\centering
\begin{tikzpicture}[scale=0.5]
\coordinate (p_0) at (3.5,2.3);
\coordinate[label=above:{$p_1$}] (p_1) at (3.5,6.0621);
\coordinate[label=below left:{$p_2$}] (p_2) at (0,0);
\coordinate[label=below right:{$p_3$}] (p_3) at (7,0);

\coordinate[label=below:{\footnotesize{$p_{11}$}}]  (p_11) at (3.5/2+3.5/2,6.0621/2+2.3/2);
\coordinate (p_12) at (3.5/4+3.5/2,6.0621/4+2.3/2);
\coordinate (p_13) at (7/4+3.5/4+3.5/2,6.0621/4+2.3/2);
\coordinate[label=above:{\footnotesize{$p_{22}$}}] (p_22) at (3.5/2,2.3/2);
\coordinate (p_23) at (7/4+3.5/2,2.3/2);
\coordinate[label=above:{\footnotesize{$p_{33}$}}] (p_33) at (7/2+3.5/2,2.3/2);
%
\draw (p_1) -- (p_2) node[midway, left] {\small{$e_{12}$}} -- (p_3) node[midway,below] {\small{$e_{23}$}} -- (p_1) node[midway, right] {\small{$e_{13}$}};
\draw (p_1) -- (p_11) node[midway, right] {\small{$\!\!e_{11}$}};
\draw (p_2) -- (p_22) node[midway, right] {\small{$\,e_{22}$}};
\draw (p_3) -- (p_33) node[midway, left] {\small{$e_{33}$}};
\fill (p_1) circle (2pt);
\fill (p_2) circle (2pt);
\fill (p_3) circle (2pt);
\fill (p_11) circle (2pt);
\fill (p_22) circle (2pt);
\fill (p_33) circle (2pt);
\end{tikzpicture}
\caption{Line segments $e_{ij}$ with $(i,j)\in B$.}
\label{fig segments}
\end{figure}

\begin{proposition}\label{prop existence Qalpha}
For any $\alpha\in (0,1)$ there exists a unique compact set $Q_{\alpha}\subseteq\mbbR^2$ such that
\[
Q_{\alpha} = \bigcup_{i\in S}G_i(Q_{\alpha}) \cup \bigcup_{(i,j)\in B} e_{ij}.
\]
Furthermore,
\begin{equation}\label{eq decomp Qalpha}
Q_{\alpha} = C_{\alpha} \cup \bigcup_{(w,(i,j)) \in W_*\times B}e_{ij}^w,
\end{equation}
where $C_{\alpha}$ is the self-similar set associated with $\{G_1, G_2, G_3\}$, i.e. $C_{\alpha}$ is the unique nonempty compact set satisfying
\begin{equation}\label{eq Kalpha}
C_\alpha = \bigcup_{i\in S}G_i(C_\alpha).
\end{equation}
\end{proposition}
\begin{proof}
The mapping $H(K):=\bigcup_{i=1}^3 G_i(K)$ is a $\frac{\alpha}{2}$-contraction. By~\cite[Theorem 1]{Hat85}, the inhomogeneous equation 
$x=H(x)\cup \cup_{(i,j)\in B} e_{ij}$ has a unique solution $Q_\alpha$ in the space of compact subsets of $\mbbR^2$ that equals the closure of 
$\bigcup_{(w,(i,j)) \in W_*\times B}e_{ij}^w$ and in particular~\eqref{eq Kalpha} holds.
\end{proof}
Notice that 
$C_\alpha$ is a fractal (Cantor) dust for any $\alpha\in(0,1)$. In view of the next proposition, we will refer to any of $Q_\alpha$ and $C_\alpha$ with $\alpha\in(0,1)$ simply by $Q_\infty$ and $C_\infty$. Although the notation might seem misleading at first sight, we write $Q_\infty$ to underline its relation with $\G_\infty$, not meaning $\alpha=\infty$.
\begin{proposition}\label{prop Qalpha homeo}
The sets $Q_\alpha$ are pairwise homeomorphic for any $\alpha\in (0,1)$.
\end{proposition}
\begin{proof}
Let $G_w^{\alpha}$ and $e^{\alpha,w}_{ij}$ denote $G_w$ and $e^w_{ij}$ respectively. Moreover, let $\iota^\alpha$ denote the canonical coding mapping that identifies $S^{\mbbN}$ with $C_\alpha$, which is given by $\iota^\alpha(w_1w_2\ldots)=\bigcap_{k\geq 1}G_{w_1\ldots w_k}(C_\alpha)$. For any $\alpha_1,\alpha_2\in (0,1)$, define $\varphi_{\alpha_1,\alpha_2}=\iota^{\alpha_2}{\circ}(\iota^{\alpha_1})^{-1}\colon C_{\alpha_1}\to C_{\alpha_2}$. Extending this mapping onto each $e^{\alpha_1,w}_{ij}$ by $\varphi_{\alpha_1,\alpha_2}|_{e^{\alpha_1,w}_{ij}}=G_w^{\alpha_2}{\circ}(G_w^{\alpha_1})^{-1}|_{e^{\alpha_1,w}_{ij}}$ for any $(w,(i,j))\in W_*\times B$ yields the desired homeomorphism $\varphi_{\alpha_1,\alpha_2}\colon Q_{\alpha_1}\to Q_{\alpha_2}$.
\end{proof}

The fractals $Q_\infty$, $C_\infty$, and the infinite graph $\G_\infty$ are related through a projection mapping $\pi\colon\G_\infty\to\mbbR^2$ that in particular maps each approximating graph $\G_n$ to its graphical representation in $\mbbR^2$ displayed in Figure~\ref{fig graphs G_n}. Based on this construction, each vertex $x\in V_n\setminus V_{n-1}$ will be associated with a word of length $n\geq 1$, $w(x)\in W_n$, so that $\pi(x)=G_{w(x)_1\ldots w(x)_{n-1}}(p_{w(x)_n})$. In view of~\eqref{eq def tilde V}, any accumulation point, which corresponds to a vertex not captured by $V_*$, will be associated with an infinite word $w(x)\in S^{\mbbN}$ provided by the canonical coding mapping associated with $C_\infty$, so that $\pi(x)=\cap_{n\geq 1}G_{w(x)_1\ldots w(x)_n}(C_\infty)$.

\medskip

\begin{definition}\label{def proj}
Let $V_\infty$ and $E_\infty$ the set of vertices, respectively edges, of $\G_\infty$. For any fixed choice of the values of $\proj|_{V_0}$ so that $\proj(V_0)=\widetilde{V}_0$, the projection mapping $\pi\colon\G_\infty\to\mbbR^2$ is defined as
\[
\proj(x)=\left\{\begin{array}{ll}
G_{w(x)}(p_{w(x)_n})&\text{if }x\in V_n\setminus V_{n-1},\\
\bigcap_{k\geq 1}G_{w_1(x)\ldots w_k(x)}(C_\infty)&\text{if }x\in V_\infty\setminus V_*,
\end{array}\right.
\]
and
\begin{align*}
\pi(\{x,y\})&=\{\proj(x)(1-t)+t\proj(y),~t\in (0,1)\},\qquad  \{x,y\}\in E_\infty.
\end{align*} 
\end{definition}
An each level $n\geq 1$, $\proj(\G_n)$ is isomorphic to the so-called \textit{cable system} associated with the graph $\G_n$, see~\cite{BB04}.

\begin{remark}\label{rem Qalpha}
The sequence $\{\pi(\G_n)\}_{n\geq 0}$ is monotonically increasing and
\[
Q_\infty=\overline{\bigcup_{n\geq 0}\bigcup_{\{x,y\}\in E_n}\proj(\{x,y\})}^{\text{Eucl}}.
\]
\end{remark}

This last observation leads to the fact that $Q_\infty$ is a \textit{fractal quantum graph}, a concept introduced in~\cite{AKT16}, whose definition we recall below.
\begin{definition}
A fractal quantum graph with length system $\{(\phi_k,\ell_k)\}_{k\geq 1}$ is a se\-parable compact connected and locally connected metric space $(X,d)$ that satisfies the following two conditions:
\begin{itemize}[leftmargin=.25in]
\item[(i)] For each $k\geq 1$, $\ell_k>0$ and $\phi_k\colon [0,\ell_k]\to X$ is an isometry such that $\phi_k([0,\ell_k])\cong[0,\ell_k]$ and
\[
\phi_k((0,\ell_k))\,\cap\,\phi_j((0,\ell_j))=\emptyset\qquad \forall~k\neq j.
\]
\item[(ii)] The set
\[
K:=X\setminus\bigcup_{k\geq 1}\phi_k((0,\ell_k))
\]
is totally disconnected.
\end{itemize}
\end{definition}
In view of the definition of $\proj$, the totally disconnected set $K_\infty$ associated with $Q_\infty$ corresponds to the closure of the set of nodes $V_*$. Potentials in the F-S ladder are thus defined on $K_\infty$ and we will see in the next section how to extend our previous definition of power dissipation to a special class of them.
\begin{proposition}\label{lemma Q is FQG}
$Q_\infty$ is a fractal quantum graph.
\end{proposition}
\begin{proof}
Equipped with the Euclidean metric, $Q_\infty$ is a compact, and thus locally compact metric space. The length system is given by $\{(\phi_{ij}^w,\ell_{ij}^w)\}_{(w,(i,j))\in W_*\times B}$, where $\ell_{ij}^w=\diam e_{ij}^w$ and $\phi_{ij}^w\colon[0,\ell_{ij}^w]\to e_{ij}^w$ is the curve parametrization of $e_{ij}^w$. In particular, $\phi_{ij}^w((0,\ell_{ij}^w))=\mathring{\pi}(\{x,y\}):=\pi(\{x,y\})\setminus\{\proj(x),\proj(y)\}$, where $\proj(x)=G_w(p_i)$ and $\proj(y)=G_w(p_j)$. In view of~\eqref{eq decomp Qalpha} we have that 
\[
K_\infty:=Q_\infty\setminus\bigcup_{(w,(i,j))\in W_*\times B}\phi_{ij}^w((0,\ell_{ij}^w))=Q_\infty\setminus\bigcup_{n\geq 0}\bigcup_{\{x,y\}\in E_n}\mathring{\proj}(\{x,y\})=C_\infty\cup \widetilde{V}_*
\]
is a totally disconnected set.
\end{proof}

Notice that $Q_\infty$ is also a \textit{finitely ramified fractal}~\cite{Tep08} and it can be expressed as a \textit{graph directed fractal}~\cite{HN03} as well. 
In our particular case we have focused on the fact that it is a fractal quantum graph because of the importance of the totally disconnected set $K_\infty$ in the next sections.
\section{Continuity of potentials}\label{setion cp in SL}
The projection mapping $\proj$ allows us to identify the nodes of the F-S ladder network with a subset of $\mbbR^2$. In general, this kind of identification naturally transfers the notion of power dissipation in graphs to discrete subsets of $\mbbR^2$. The harmonic functions associated with power dissipation and in particular their continuity, proved in Theorem~\ref{thm harmonic cont}, are essential to define power dissipation in $Q_\infty$ through the fractal dust $K_\infty$.

\medskip

From now on, we identify the approximating sets $\widetilde{V}_n$ in~\eqref{eq def tilde V} with the set of vertices $V_n$ via $\widetilde{V}_n=\proj(V_n)$ and hence use the notation $V_n$ for any of both. The power dissipation in $V_n$ associated with a network $\net$ (formally given by $\pd{\net}[v{\circ}\pi]$) will be denoted by $\pdf{\net}[v]$.
%
%
%
In this manner, the power dissipation associated with the F-S ladder is given by 
\begin{equation*}\label{eq def pdinWstar}
\pdf{\SL}[v]:=\lim_{\varepsilon\to 0_+}\lim_{n\to\infty}\pdf{\net_{\varepsilon,n}}[v_{|_{V_n}}],
\end{equation*}
where $\dom \pdf{\SL}:=\{v\in\ell(V_*)~|~\pdf{\SL}[v]<\infty\}$ and now $V_*$ is a subset of $\mbbR^2$.

\medskip

As already mentioned, Lemma~\ref{lemma Q is FQG} and the definition of $\proj$ allow us to identify the set of vertices $V_\infty$ with the totally disconnected set $K_\infty$ for which $V_*$ is a dense subset. By Proposition~\ref{prop existence Qalpha}, $K_\infty$ is compact with respect to the topology induced by the Euclidean metric. The aim of this section is to prove the continuity of the harmonic functions on $V_*$, so that they can be uniquely extended to continuous (harmonic) functions on $K_\infty$. 

\medskip

Recall that a function $h\in\ell(V_*)$ is said to be \textit{harmonic} if for any $\varepsilon>0$ and $n\geq 0$
\[
\pdf{\net_{\varepsilon,0}}[h_{|_{V_0}}]=\pdf{\net_{\varepsilon,n}}[h_{|_{V_n}}].
\]
Moreover, the space of harmonic functions on $V_*$, denoted by $\Hm_{\SL}(V_*)$, is $3$-dimensional, and for any $h\in\Hm_{\SL}(V_*)$ and $n\geq 1$,
\begin{equation}\label{eq PSL vs PZen}
\lim\limits_{\varepsilon\to 0_+}\pdf{\net_{\varepsilon,n}}[h_{|_{V_n}}]=\pdf{\SL}[h]=\pdf{\net_n}[h_{|_{V_n}}]
\end{equation}
c.f.  Proposition~\ref{prop harmonic in V}.

\medskip

Starting with a function $h_0\in\ell(V_0)$, harmonic functions are constructed by applying recursively the \textit{harmonic extension algorithm} provided in~\cite[Theorem 3.3]{A++16}. This result conveys an explicit expression of the $3\times 3$-matrices $A_1,A_2,A_3$, that describe the algorithm. Therefore,
\begin{equation}\label{eq harmonic matrices}
h_{|_{G_j(V_0)}}=A_j h_{|_{V_0}}. 
\end{equation}
for any $h\in\Hm_{\SL}(V_*)$ and $j=1,2,3$.

\begin{remark}\label{rem evs of Aj}
\begin{itemize}[leftmargin=.25in]
\item[(i)]The eigenvalues of $A_j$, $j=1,2,3$, can be explicitly computed with any mathematical software and equal
\[
\lambda_1=1,\quad \lambda_2=\frac{3\Zeff{\SL}}{9Z_C+5\Zeff{\SL}},\quad\lambda_3=\frac{1}{3}\lambda_2.
\]
\item[(ii)] While the eigenvector associated with $\lambda_1$ is $h_1=(1,1,1)$ in all matrices, the eigenvectors associated with $\lambda_2$ and $\lambda_3$ vary with the choice of $j$. The space of constant functions on $V_*$ is thus the $1$-dimensional subspace of $\Hm_{\SL}(V_*)$ spanned by $h_1$.

\item[(iii)] Under the filter condition~\eqref{FC}, substituting $Z_C$ and $\Zeff{\SL}$ by their actual values from~\cite[Theorem 3.1]{A++16}, one obtains 
\[
|\lambda_2|^2=\frac{9\sigma^2+(27+6CL\omega^2)^2}{2106+25\sigma^2+90\sigma+100CL\omega^2(9+2CL\omega^2)},
\]
where $\sigma=\sqrt{144 CL\omega^2-(2CL\omega^2)^2-81}\in\mbbR$. Although not directly readable from this expression, it holds that $|\lambda_2|<1$. 
\end{itemize}
\end{remark}

\begin{lemma}\label{lemma eigenvalue less than 1}
The eigenvalues of the matrices $A_1,A_2,A_3$ from the harmonic extension algorithm satisfy $|\lambda_3|<|\lambda_2|<|\lambda_1|=1$.
\end{lemma}
\begin{proof}
In view of~\ref{rem evs of Aj}(i) it only remains to prove that $|\lambda_2|<1$. Let us consider for instance the matrix $A_1$ and let $h_2$ be the eigenvector associated with $\lambda_2$. Further, let $h\in\Hm_{\SL}(V_*)$ be the harmonic function with $h_{|_{V_0}}=h_2$ and denote by $D^2_0$ the matrix representation of the power dissipation $\pdf{\net_{0}}$, i.e. 
\begin{equation}\label{eq def QP}
D^2_0=\frac{\Re(\Zeff{\SL})}{2|\Zeff{\SL}|^2}\begin{pmatrix}2&-1&-1\\-1&2&-1\\-1&-1&2\end{pmatrix}.
\end{equation}
Since $h$ is harmonic, we have that
\begin{align*}
\pdf{\net_0}[A_1 h_{|_{V_0}}]&=\langle D^2_0A_1 h_2,A_1h_2\rangle=|\lambda_2|\langle D^2_0h_2,h_2\rangle\\
&=|\lambda_2|\pdf{\net_0}[h_{|_{V_0}}]=|\lambda_2|\pdf{\SL}[h].
\end{align*}
On the other hand, it follows from~\eqref{eq PSL vs PZen} and the definition of $\pdf{\net_n}$ that 
\begin{equation}\label{eq self sim eq harmonic}
\pdf{\SL}[h]=\pdf{\net_1}[h_{|_{V_1}}]=\sum_{j=1}^3\pdf{\net_{0}}[A_j h_{|_{V_0}}].
\end{equation}
Thus, if $|\lambda_2|=1$, then $\pdf{\net_{0}}[A_2 h_{|_{V_0}}]=\pdf{\net_0}[A_3 h_{|_{V_0}}]=0$ and hence $A_2h_2$ and $A_3h_2$ are constant, a contradiction. 
\end{proof}

\begin{remark}
In fact, it is possible to check directly that for instance $A_3h_2$ is non-constant because~\cite[Theorem 3.3]{A++16} provides the explicit expression of $A_3$ and $h_2$, which leads to
\[
A_3h_2=\frac{3\Zeff{\SL}}{9Z_c+5\Zeff{\SL}}\bigg(3,\frac{27Z_C+10\Zeff{\SL}}{3Z_C+2\Zeff{\SL}},\frac{18Z_C+8\Zeff{\SL}}{3Z_C+2\Zeff{\SL}}\bigg).
\]
If this were to be constant, then $27Z_C+10\Zeff{\SL}=18Z_C+8\Zeff{\SL}$, equivalently $\Zeff{\SL}=-\frac{9}{2}Z_C$. But $Z_C$ is purely imaginary, whereas $\Zeff{\SL}$ has positive real part, a contradiction.
\end{remark}

The next theorem is the main result of this section. It justifies the extension of power dissipation in the F-S ladder to potentials defined on the whole $K_\infty$.

\begin{theorem}\label{thm harmonic cont}
Harmonic functions are continuous on $V_*$, i.e. $\Hm_{\SL}(V_*)\subseteq C(V_*)$.
\end{theorem}

Before proving this result we show the following key lemma.

\begin{lemma}\label{lem contract P_0}
There exists $r\in(0,1)$ such that for any non-constant $h_0\in\ell(V_0)$ 
\[
\pdf{\net_0}[A_jh_0]\leq r^2\pdf{\net_0}[h_0],\qquad\quad j=1,2,3.
\]
\end{lemma}
\begin{proof}
Let $h_0\in\ell(V_0)$ be non-constant and let $D^2_0$ be the matrix representation~\eqref{eq def QP} of $\pdf{\net_0}$. Since $D^2_0$ is non-negative definite and symmetric, $\langle D^2_0A_jh_0,h_0\rangle\geq 0$ for any $j=1,2,3$. Consider $j$ arbitrary but fixed.

Let $h_2$ and $h_3$ denote the eigenvectors of $A_j$ associated with the eigenvalues $\lambda_2$, resp. $\lambda_3$ given in Remark~\ref{rem evs of Aj} (i). Non-constant harmonic functions are thus the $2$-dimensional subspace of $\Hm_{\SL}(V_*)$ spanned by $\{h_2,h_3\}$ and hence $h_0=\sum_{k=2}^3a_kh_k$, $a_k\in\BbC$. Then,
\begin{align*}
\pdf{\net_0}[A_jh_0]&=\langle D^2_0A_jh_0,A_j h_0\rangle=\bigg|\sum_{k,l=2}^3a_k\overline{a_l}\langle D^2_0\lambda_kh_k,\lambda_lh_l\rangle\bigg|\\
&\leq \sum_{k,l=2}^3|\lambda_k\overline{\lambda_l}|\langle D^2_0a_kh_k,a_lh_l\rangle\leq r^2\langle D^2_0h_0,h_0\rangle
\end{align*}
with $r=|\lambda_2|<1$ in view of Lemma~\ref{lemma eigenvalue less than 1}.
\end{proof}

\begin{proof}[Proof of Theorem~\ref{thm harmonic cont}]
Without loss of generality, let $h\in\Hm_{\SL}(V_*)$ be non-constant and such that $\pdf{\SL}[h]=1$. For each $\varepsilon >0$ and $m\geq 0$ large enough, points inside a $m$-cell $G_w(V_*)$, $w\in W_m$, satisfy $|x-y|<\delta$ for some $\delta>0$. Since $h$ is harmonic, the maximum principle guarantees that $h$ takes its maximum and minimum value within $G_w(V_*)$ on the boundary $G_w(V_0)$. Thus, for any $x,y\in G_w(V_0)$, by definition of $\net_m$ we have that
\[
|h(x)-h(y)|^2\leq\frac{2|\Zeff{\SL}|^2}{\Re(\Zeff{\SL})}\sum_{\substack{x,y\in G_w(V_0)\\ \{x,y\}\in E_m}}\pdf{\net_0}[h]_{xy}=\frac{2|\Zeff{\SL}|^2}{\Re(\Zeff{\SL})}\pdf{\net_0}[h_{|_{G_w(V_0)}}]
\]
and since $h$ is harmonic, $h_{|_{G_w(V_0)}}=A_{w_1}\cdots A_{w_m}h_{|_{V_0}}$, with $w=w_1\ldots w_m$. Applying repeatedly Lemma~\ref{lem contract P_0} yields
\[
\pdf{\net_0}[h_{|_{G_w(V_0)}}]\leq r^{2m}\pdf{\net_0}[h_{|_{V_0}}]
\]
and hence
\[
|h(x)-h(y)|\leq |\Zeff{\SL}|\sqrt{2/\Re(\Zeff{\SL})}\,r^m<\varepsilon
\]
for $m$ large.
\end{proof}

As an immediate consequence of this result, the space of harmonic functions on $K_\infty$, denoted by $\Hm_{\SL}(K_\infty)$, is well-defined. For any $h\in\Hm_{\SL}(K_\infty)$ we will identify $\pdf{\SL}[h]$ with the former $\pdf{\SL}[h_{|_{V_*}}]$ to obtain the power dissipation associated with the F-S ladder for harmonic potentials on the fractal dust $K_\infty$.

\begin{definition}\label{def pd SL}
The power dissipation in $K_\infty$ associated with $\net_{\SL}$ of a function $h\in\Hm_{\SL}(K_\infty)$ is given by
\[
\pdf{\SL}[h]=\lim_{\varepsilon\to 0_+}\lim_{n\to\infty}\pdf{\net_{\varepsilon,n}}[h_{|_{V_n}}].
\]
\end{definition}

\section{Continuity of the power dissipation measure}\label{section cpdm}
In analogy to energy measures, this section aims to construct a measure on $K_\infty$ that can be understood as the ``power dissipation measure'' associated with harmonic potentials. The main theorem states the existence of this continuous (atomless) measure.

\begin{theorem}\label{thm pd is a measure}
For each non-constant harmonic function $h\in\Hm_{\SL}(K_\infty)$, the power dissipation $\pdf{\SL}$ induces a continuous measure $\nu_h$ on $K_\infty$ with $\supp\nu_h=C_\infty$.
\end{theorem}

Before proving this result, we provide some useful observations. In this and the next section, $n$-cells will be denoted by $T_w=G_w(K_\infty)$ for any $w\in W_n$, $n\geq 0$. 
\begin{remark}\label{rem obs for nu_h}
The following hold.
\begin{itemize}[leftmargin=.3in]

\item[(i)] For any $h\in\Hm_{\SL}(K_\infty)$ and $\{x,y\}\in E_m$, $m\geq 0$,
\[
\lim_{\varepsilon\to 0_+}\lim_{n\to\infty}\pdf{\net_{\varepsilon,n}}[h]_{xy}=\lim_{\varepsilon\to 0_+}\pdf{\net_{\varepsilon,m}}[h]_{xy}=0.
\]
\medskip

\item[(ii)] For any $h\in\Hm_{\SL}(K_\infty)$ and $w\in W_m$, $m\geq 1$,
\[
\lim_{\varepsilon\to 0_+}\lim_{n\to\infty}\sum_{\substack{x,y\in T_w\cap V_n \\ \{x,y\}\in E_n}}\pdf{\net_{\varepsilon,n}}[h]_{xy}=\frac{\Re(\Zeff{\SL})}{2|\Zeff{\SL}|^2}\sum_{x,y\in\partial T_w}|h(x)-h(y)|^2.
\] 
%
\end{itemize}
\end{remark}

\medskip

\begin{proof}[Proof of Theorem~\ref{thm pd is a measure}]
Let $h\in\Hm_{\SL}(K_\infty)$ be non-constant. For each $m$-cell $T_w$ define
\[
\nu_h(T_w):=\lim_{\varepsilon\to 0_+}\lim_{n\to\infty}\sum_{\substack{x,y\in T_w\cap V_n \\ \{x,y\}\in E_n}}\pdf{\net_{\varepsilon,n}}[h]_{xy}.
\]
Notice that since $h$ is harmonic,
\[
0\leq\nu_h(T_w)\leq\nu_h(K_\infty)=\pdf{\SL}[h]<\infty.
\]

Applying the same definition of $\nu_h$ to isolated points, we have that $\nu_h(\{x\})=0$ whenever $x\in V_m$ for some $m\geq 0$ 
because no edges are involved. If $x\in K_\infty$ is an accumulation point, it satisfies $x=\bigcap_{k\geq 1}T_{w_1\ldots w_k}$
for some infinite word $w_1w_2\ldots\in S^{\mbbN}$. In view of Remark~\ref{rem obs for nu_h}(ii) we have that
\begin{align*}
\nu_h(\{x\})&=\lim_{\varepsilon\to 0_+}\lim_{n\to\infty}\lim_{m\to\infty}\sum_{\substack{x,y\in T_{w_1\ldots w_m}\cap V_n \\ \{x,y\}\in E_n}}\pdf{\net_{\varepsilon,n}}[h]_{xy}\\
&=\lim_{\varepsilon\to 0_+}\lim_{n\to\infty}\lim_{m\to\infty}\frac{\Re(\Zeff{\varepsilon})}{2|\Zeff{\varepsilon}|^2}\sum_{\substack{x,y\in\partial T_{w_1\ldots w_m}\\ \{x,y\}\in E_n}}|h(x)-h(y)|^2.
\end{align*}
By Theorem~\ref{thm harmonic cont}, $h$ is continuous and therefore for any $\delta>0$ we find $m_0\geq 0$ large enough such that 
$\nu_h(\{x\})<\frac{3\Re(\Zeff{\SL})}{2|\Zeff{\SL}|^2}\delta^2$ for all $m\geq m_0$. Thus, $\nu_h(\{x\})=0$.

Let us now prove finite-additivity of $\nu_h$: Let $T_{w},T_{v}$ be two disjoint cells of levels $n_1$, $n_2$. If these cells can be connected by edges without additional vertices, there is at most one such connecting edge $\{p_1,p_2\}\in E_m$ with $m=\max\{n_1,n_2\}$. Let us suppose this is the case. In view of Remark~\ref{rem obs for nu_h}(i) we have 
\begin{align*}
\nu_h(T_{w}\cup T_{v})&=\lim_{\varepsilon\to 0_+}\lim_{n\to\infty}\sum_{\substack{x,y\in T_{w}\cap V_n \\ \{x,y\}\in E_n}}\pdf{\net_{\varepsilon,n}}[h]_{xy}+\lim_{\varepsilon\to 0_+}\lim_{n\to\infty}\pdf{\net_{\varepsilon,n}}[h]_{p_1p_2}\\
&+\lim_{\varepsilon\to 0_+}\lim_{n\to\infty}\sum_{\substack{x,y\in T_{v}\cap V_n \\ \{x,y\}\in E_n}}\pdf{\net_{\varepsilon,n}}[h]_{xy}=\nu_h(T_{w_1})+\nu_h(T_{w_2}).
\end{align*}
If there is no possible connecting edge, the above equality follows directly. The same argument applies for any finite union (both of cells or isolated points), since they can be connected by at most finitely many single edges.
 
In order to prove $\sigma$-additivity, consider a sequence of pairwise disjoint cells $\{T_{v(k)}\}_{k\geq 1}$, where $T_{v(k)}$ is a $n_k$-cell. Without loss of generality, assume that $n_k\leq n_{k+1}$. Notice that, since the cells are pairwise disjoint, it is only possible to connect $T_{v(k)}$ and $T_{v(k+1)}$ by an edge $\{x,y\}\in E_{n_{k+1}}$ if $n_{k+1}=n_k+1$. Hence, we can assume that there will be at most and edge $\{x_{k+1},y_{k+1}\}\in E_{n_{k+1}}$ joining $T_{v(k)}$ and $T_{v(k+1)}$. Then,
\begin{align*}
\nu_h\big(\bigcup_{k\geq 1}T_{v(k)}\big)&=\lim_{\varepsilon\to 0_+}\lim_{n\to\infty}\lim_{m\to\infty}\sum_{k=1}^m\sum_{\substack{x,y\in T_{v(k)}\cap V_n \\ \{x,y\}\in E_n}}\pdf{\net_{\varepsilon,n}}[h]_{xy}\\
&+\lim_{\varepsilon\to 0_+}\lim_{n\to\infty}\lim_{m\to\infty}\sum_{k=2}^{m+1}\pdf{\net_{\varepsilon,n_k}}[h]_{x_ky_k}=\sum_{k=1}^\infty\nu_h(T_{w(k)})+0,
\end{align*}
where last equality follows by Remark~\ref{rem obs for nu_h}(i) after interchanging the order of the limits, which is possible because both summands are uniformly bounded by $\pdf{\SL}[h]$.
The same argument applies to countable unions of isolated points, which in particular implies that $\supp \nu_h =C_\infty$.
%


Finally, by Carath\'eodory's extension theorem, $\nu_h$ admits a unique extension to a (finite) measure on $K_\infty$ and this measure is continuous because isolated points have no mass.
%
\end{proof}

\begin{corollary}
For any $m\geq 0$, $w\in W_m$, and any $m$-cell $T_w$ it holds that
\[
\nu_h(T_w)\asymp\osc(h_{|_{T_w}})^2,
\]
where $\osc(h_{|_{T_w}})=\max\limits_{x\in T_w}h(x)-\min\limits_{y\in T_w}h(y)$.
\end{corollary}
\begin{proof}
First of all, recall from Remark~\ref{rem obs for nu_h}(ii) that
\[
\nu_h(T_w)=\frac{\Re(\Zeff{\SL})}{2|\Zeff{\SL}|^2}\sum_{x,y\in\partial T_w}|h(x)-h(y)|^2.
\]
By the maximum principle and since $h$ is harmonic, $h_{|_{T_w}}$ takes its maximum and minimum on the boundary $\partial T_w$. Hence, $\osc(h_{|_{T_w}})=|h(x)-h(y)|$ for some $x,y\in\partial T_w$ and the definition of $\nu_h$ yields
\[
\frac{\Re(\Zeff{\SL})}{2|\Zeff{\SL}|^2}\osc(h_{|_{T_w}})^2\leq\nu_h(T_w)\leq\frac{3\Re(\Zeff{\SL})}{2|\Zeff{\SL}|^2}\osc(h_{|_{T_w}})^2.
\]
\end{proof}
\section{Singularity of the power dissipation measure}
This section is devoted to proving that the power dissipation measure $\nu_h$ discussed in the preceding section is singular with respect to the Bernouilli measure $\mu$ on $K_\infty$ that satisfies
\[
\mu(T_{w_1\ldots w_n})=\mu_{w_1}\cdots\mu_{w_n}
\]
for any $n$-cell $T_{w_1\ldots w_n}$, $w_1\ldots w_n\in W_n$, where $\sum_{i\in S}\mu_i=1$. In particular in this case, we will consider $\mu_1=\mu_2=\mu_3=\frac{1}{3}$. Together with the measure $\mu$, $K_\infty$ can be seen as a probability space. Notice that, as it happened with $\nu_h$, $\supp\mu =C_\infty$, where $C_\infty$ is the Cantor dust defined in~\eqref{eq Kalpha}.

\medskip

Recall that any element in the support of $\mu$ is a non-isolated point of $K_\infty$ such that $x=\bigcap_{n\geq 1}T_{w_1\ldots w_n}$ for some $w_1w_2\ldots\in S^{\mbbN}$. For these points, we define the (random) matrices $M_n(x)=A_{w_n}$, where $A_j$, $j\in S$, are the matrices of the harmonic extension algorithm~\eqref{eq harmonic matrices}. The matrices $M_n(x)$ are statistically independent with respect to $\mu$.

\medskip

The next result is based on a special case of~\cite{Kus89} and we will mainly follow the proof given in~\cite[Theorem 5.1]{BST99}, including details for completeness. Since we are only dealing with non-constant harmonic functions, we will restrict to the $2$-dimensional subspace of $\Hm_{\SL}(K_\infty)$ spanned by the two harmonic functions $h_2,h_3$ associated to the eigenvalues $\lambda_2,\lambda_3$ from Remark~\ref{rem evs of Aj}(i). 

\begin{theorem}\label{thm singularity of pd}
Assume that for a non-constant $h\in\Hm_{\SL}(K_\infty)$ there exists $m\geq 1$ such that the mapping $x\mapsto \|D_0M_m(x)\cdots M_1(x)h_{|_{V_0}}\|$ is non-constant. Then, $\nu_h$ is singular with respect to $\mu$.
\end{theorem}

The proof of this theorem essentially consists in proving the condition stated in the following lemma, which is a consequence of the generalized Lebesgue differentiation theorem for metric measure spaces.
\begin{lemma}\label{prop Leb diff thm}
The measure $\nu_h$ is singular with respect to $\mu$ if for $\mu$-a.e. $x\in C_\infty$
\[
\lim_{n\to\infty}\frac{\nu_h(T_{w_1\ldots w_n})}{\mu(T_{w_1\ldots w_n})}=0,
\]
where $x=\bigcap_{n\geq 1}T_{w_1\ldots w_n}$.
\end{lemma}
\begin{proof}
For each $n\geq 1$, $T_{w_1\ldots w_n}$ is a neighborhood of $x$ and $\lim_{n\to\infty}\mu(T_{w_1,\ldots w_n})=0$. Let us suppose that $\nu_h$ is absolutely continuous with respect to $\mu$. Then, there exists a measurable function $f$ (the Radon-Nikodym derivative of $\nu_h$ with respect to $\mu$) such that
\begin{equation}\label{eq RN derivative}
\frac{\nu(T_{w_1\ldots w_n})}{\mu(T_{w_1\ldots w_n})}
=\frac{1}{\mu(T_{w_1\ldots w_n})}\int_{T_{w_1\ldots w_n}}f(x)\,\mu(dx).
\end{equation}
Due to the definition of $\mu$ and since $C_\infty$ is self-similar, it is Ahlfors regular (i.e. $\mu(B_{d_E}(x,r)\cap C_\infty)\asymp r^{\gamma}$, in this case with $\gamma$ being the Hausdorff dimension of $C_\infty$). Thus, $(C_\infty,\mu)$ equipped with the Euclidean metric is volume doubling 
and the generalized Lebesgue differentiation theorem, see e.g.~\cite[Theorem 1.8]{Hei01} holds, so that~\eqref{eq RN derivative} equals $f(x)$ for $\mu$-a.e. $x\in C_\infty$. By assumption, this implies that $f$ is zero $\mu$-a.e., a contradiction.
\end{proof}

\begin{proof}[Proof of Theorem~\ref{thm singularity of pd}]
For any $h\in\Hm_{\SL}(K_\infty)$, $w\in W_m$ and $n$ large it holds that
\begin{align*}
\sum_{\substack{x,y\in T_w\cap V_n\\ \{x,y\}\in E_n}}\pdf{\net_{\varepsilon,n}}[h]_{xy}&= \sum_{\substack{x,y\in T_w\cap V_{n-m}\\ \{x,y\}\in E_n}}\pdf{\net_{\varepsilon,n}}[h\circ G_w]_{xy}\\
&=\sum_{\substack{x,y\in T_w\cap V_{n-m}\\ \{x,y\}\in E_{n-m}}}\pdf{\net_{\varepsilon,n-m}}[h\circ G_w]_{xy}=\pdf{\net_{\varepsilon,n-m}}[h\circ G_w{}_{|_{V_{n-m}}}]\\
&=\frac{\Re(\Zeff{\varepsilon})}{2|\Zeff{\varepsilon}|^2}\sum_{\{x,y\}\in E_0}|h\circ G_w(x)-h\circ G_w(y)|^2.
\end{align*}
Letting $\varepsilon\to 0_+$ and $n\to\infty$ in both sides of the equality yields
\begin{align}
\nu_h(T_w)&=\frac{\Re(\Zeff{\SL})}{2|\Zeff{\SL}|^2}\sum_{\{x,y\}\in E_0}|h\circ G_w(x)-h\circ G_w(y)|^2\nonumber\\ 
&=\|D_0A_{w_m}\cdots A_{w_1}h_{|_{V_0}}\|^2.\label{eq BST 5.2}
\end{align}

\medskip

On the other hand, it follows from~\eqref{eq self sim eq harmonic} and the definition of $\mu$ that 
\begin{align}
\|D_0h_{|_{V_0}}\|^2&=\langle D^2_0h_{|_{V_0}},h_{|_{V_0}}\rangle=\pdf{\net_0}[h_{|_{V_0}}]=\sum_{i=1}^3\pdf{\net_0}[h\circ G_i{}_{|_{V_0}}]\nonumber\\
&=\sum_{i=1}^3\langle D^2_0A_ih_{|_{V_0}},A_ih_{|_{V_0}}\rangle=\sum_{i=1}^3\|D_0A_ih_{|_{V_0}}\|^2\nonumber\\
&=3\sum_{i=1}^3\mu_i\|D_0A_ih_{|_{V_0}}\|^2
=3\int_{K_\infty}\|D_0M_1(x)h_{|_{V_0}}\|^2\mu(dx).\nonumber
\end{align}
and induction leads to
\begin{equation}\label{eq BST 5.4}
\|D_0h_{|_{V_0}}\|^2=3^n\int_{K_\infty}\|D_{P_0}M_n(x)\cdots M_1(x)h_{|_{V_0}}\|^2\mu(dx)
\end{equation}
for any $n\geq 1$. Notice that all computations are in fact basis independent.

\medskip

Furthermore, by assumption, there exists $m\geq 1$ such that Jensen's and the Cauchy-Schwartz inequality yield
\begin{align*}
&\int_{K_\infty}\log\|D_0M_m(x)\cdots M_1(x)h_{|_{V_0}}\|\,\mu(dx)<\log \int_{K_\infty}\|D_0M_m(x)\cdots M_1(x)h_{|_{V_0}}\|\,\mu(dx)\\
&\leq\frac{1}{2}\log \int_{K_\infty}\|D_0M_m(x)\cdots M_1(x)h_{|_{V_0}}\|^2\mu(dx)=\frac{1}{2}\log \frac{1}{3^m}\|D_0h_{|_{V_0}}\|^2,
\end{align*}
where last equality follows from~\eqref{eq BST 5.4}. Hence, 
\begin{equation}\label{eq BST 5.5}
\beta:=\sup_{h\in\Hm_1}\int_{K_\infty}\log\|D_0M_m(x)\cdots M_1(x)h_{|_{V_0}}\|\,\mu(dx)<-\frac{m}{2}\log 3,
\end{equation}
where $\Hm_1:=\{h\in\Hm_{\SL}(K_\infty)\text{ non-constant, }\|D_0h_{|_{V_0}}\|=1\}$.

\medskip

Let now $h\in\Hm_1$ and $n\geq 1$. Multiplying and dividing by $\|D_0M_{m(n-1)}(x)\cdots M_1(x)h_{|_{V_0}}\|$ and since the matrices $M_i(x)$ are statistically independent, Jensen's inequality yields
\begin{align*}
&\int_{K_\infty}\log\|D_0M_{mn}(x)\cdots M_{mn-m}(x)\cdots M_1(x)h_{|_{V_0}}\|\,\mu(dx)\\
&\leq\beta+\int_{K_\infty}\log\|D_0M_{m(n-1)}(x)\cdots M_1(x)h_{|_{V_0}}\|\,\mu(dx).
\end{align*}
By induction we obtain
\[
\int_{K_\infty}\log\|D_0M_{mn}(x)\cdots  M_1(x)h_{|_{V_0}}\|\,\mu(dx)\leq n\beta
\]
and thus
\begin{equation}\label{eq BST 5.6}
\int_{K_\infty}\log\|D_0M_{m_n}(x)\cdots  M_1(x)h_{|_{V_0}}\|\,\mu(dx)\leq \frac{m_n}{m}\beta
\end{equation}
for the subsequence $m_n=mn$. Consequently,
\begin{equation}\label{eq bound limsup}
\limsup_{m_n\to\infty}\frac{1}{m_n}\log\|D_0M_{m_n}(x)\cdots  M_1(x)h_{|_{V_0}}\|\leq\frac{1}{m}\beta<-\frac{1}{2}\log 3,
\end{equation}
where the existence of the limit is guaranteed by Furstenberg's Theorem~\cite{Fur73} because the matrices $M_{m_n}(x)$ are i.i.d. By definition of $\mu$ and~\eqref{eq BST 5.2} we thus have that
\[
\frac{\nu_h(T_{w_1\ldots w_n})}{\mu(T_{w_1\ldots w_n})}=3^n\|D_0M_{w_n}(x)\cdots M_{w_1}(x)h_{|_{V_0}}\|^2
\]
for $\mu$-a.e. $x\in C_\infty$, hence $\mu$-a.e. in $K_\infty$, and~\eqref{eq BST 5.6} yields
\begin{align*}
\frac{1}{n}\big(\log \|D_0M_{w_n}(x)\cdots M(x)_{w_1}h_{|_{V_0}}\|^2+n\log 3\big)<-n\log 3-n\log 3=0.
\end{align*}
Finally, this implies that
\begin{equation}\label{eq BST 3.7}
\lim_{n\to\infty}\frac{\nu_h(T_{w_1\ldots w_n})}{\mu(T_{w_1\ldots w_n})}=0
\end{equation}
for $\mu$-a.e. $x\in K_\infty$. By Proposition~\ref{prop Leb diff thm}, $\nu_h$ is singular with respect to $\mu$.
\end{proof}
\subsection*{Acknowledgments}
The author would like to thank A. Teplyaev and L. Rogers for very fruitful discussions.
\bibliographystyle{amsplain}
\bibliography{RefsACC}
\end{document}